\documentclass[journal]{IEEEtran}

\usepackage[T1]{fontenc}
\usepackage{amsmath}
\usepackage{amsthm}
\usepackage[cmintegrals]{newtxmath}
\usepackage{amssymb}
\usepackage{bm}
\usepackage{bbm}
\usepackage{comment}
\usepackage{graphicx}
\usepackage{calc}
\usepackage{epstopdf}
\newtheorem{theorem}{Theorem}

\usepackage{verbatim}
\usepackage{color}
\usepackage{array}
\usepackage{cite}
\usepackage{stfloats}

\usepackage{xcolor}
\usepackage{floatrow}
\usepackage{lipsum}
\usepackage[ruled,norelsize]{algorithm2e}
\usepackage{algorithmic}
\usepackage[caption=false,subrefformat=parens,labelformat=parens]{subfig}
\usepackage{multirow}
\begin{document}
\title{Learning Based Hybrid Beamforming  for Millimeter Wave Multi-User  MIMO Systems}

\author{
	Shaocheng~Huang, Yu~Ye, \IEEEmembership{Student Member, IEEE}  and  Ming~Xiao, \IEEEmembership{Senior Member, IEEE}
	\thanks{S.~Huang, Y.~Ye and M.~Xiao are with the Division of Information Science and Engineering, KTH Royal Institute of Technology, Stockholm, Sweden (e-mail: \{shahua, yu9, mingx\}@kth.se).}
}

\maketitle
 
\begin{abstract}
 
 Hybrid beamforming (HBF) design is a crucial stage in millimeter wave (mmWave) multi-user multi-input multi-output (MU-MIMO) systems. However, conventional HBF methods are still with high complexity and strongly rely on the quality of channel state information. We propose an extreme learning machine (ELM) framework to jointly optimize transmitting and receiving beamformers. Specifically, to provide accurate labels for training, we first propose an factional-programming and majorization-minimization based HBF method (FP-MM-HBF). Then, an ELM based HBF (ELM-HBF) framework is proposed to increase the robustness of beamformers. Both FP-MM-HBF and ELM-HBF can provide higher system sum-rate compared with existing methods. Moreover, ELM-HBF cannot only provide robust HBF performance, but also consume very short computation time.

\end{abstract}

\IEEEpeerreviewmaketitle

\section{Introduction}

\IEEEPARstart{R}{cently}, hybrid  beamforming (HBF) design for millimeter wave (mmWave)   multi-user multi-input  multi-output (MU-MIMO) systems has been receiving  increasing research attention due to the advantages of using less fully digital (FD) beamformers and providing high beamforming gain to overcome the severe pathloss at affordable hardware cost \cite{Ming2017xiao,huang2020learning,nguyen2017hybrid,ni2015hybrid}.  
Generally, the main challenges  of the HBF optimization problem for mmWave MU-MIMO systems  are the non-convex constraints of  analog beamformers and the inter-user interference. 
Some optimization based methods have been proposed to optimize the hybrid beamformers by using following FD methods: block diagonalization zero-forcing (BD-ZF) \cite{spencer2004zero,ni2015hybrid}, minimum mean square error (MMSE) and  weighted MMSE (WMMSE)\cite{nguyen2017hybrid}. 
For example, to design hybrid beamformers, a WMMSE based orthognal matching pursuit (OMP) method is proposed   in \cite{nguyen2017hybrid} and a BD-ZF based exhaustive search method is proposed in \cite{elbir2019hybrid}. 
However, above approaches  either require perfect channel state information (CSI), or  high computational capability.

Recent development in   machine learning (ML) provides a new way for addressing problems in physical layer communications (e.g., direction-of-arrival estimation\cite{huang2018deep}, HBF design  \cite{huang2019deep,elbir2019hybrid} and channel estimation\cite{ye2017power}). ML based techniques have several advantages such as low complexity when solving non-convex problems and the ability  to extrapolate new features from noisy and limited training data \cite{elbir2019cnn}.  
In \cite{elbir2019hybrid},   a convolutional neural network (CNN) framework is first proposed to optimize   hybrid beamformers  for mmWave MU-MIMO systems, in which the network takes the imperfect channel matrix as the  input and produces the analog and digital  beamformers  as outputs. However, this work is only feasible for single stream transmission and the  exhaustive search HBF  algorithm is with extremely  high computational complexity.
 Furthermore, methods in \cite{huang2019deep,elbir2019hybrid}  using multiple large dimensional layers constructions  may consume tremendous computation  time in training phase, which is impractical with the hardware constraint (e.g., limited computational capability and memory resources) of mobile terminals.  


From above observations,  we propose an   extreme learning machine (ELM)  framework with easy implementation to jointly optimize transmitting  and  receiving beamformers. The main contributions are summarized as follows: 

 \begin{itemize}
 \item \textbf{New HBF optimization  algorithms}: We decouple the HBF optimization problem to two sub-problems.  Unlike the work in \cite{nguyen2017hybrid} that minimizes MMSE, we first propose an factional programming (FP) based  FD beamforming algorithm that directly  maximizes the system sum-rate. With FD beamformers, we then propose a low-complexity majorization-minimization (MM)  based algorithm to jointly optimize  hybrid beamformers. Finally, the convergence  and computational complexity of  proposed algorithms are analyzed. We show that above proposed methods can achieve higher sum-rate than conventional FD beamforming and  HBF methods.
 
  \item \textbf{Robust and low-complexity ML based HBF design}: We  propose an  easily   implemented  ML based HBF framework (i.e.,  ELM-HBF) to jointly estimate the precoders and combiners. Different from conventional HBF methods, of which the performance strongly relies on the quality of CSI, our learning based approach can achieve more robust performance since ELM  are effective at handling the imperfections and  corruptions in the input channel information. We show that, for $64$ transmit antennas, ELM-HBF can achieve  50   times faster prediction time than WMMSE-OMP-HBF\cite{nguyen2017hybrid}, and  1200  times faster training time than CNN-HBF\cite{elbir2019hybrid}.
 
\end{itemize} 

\emph{Notations}:  Bold lowercase and uppercase letters denote vectors and matrices, respectively.  $\text{Tr}( {\bf A})$, $ | {\bf A}  |$, $ \| {\bf A}  \|_\text{F}$, ${\bf A}^*$, ${\bf A}^T$ and ${\bf A}^H$ denote trace, determinant, Frobenius norm, conjugate, transpose and conjugate transpose of matrix ${\bf A}$, respectively. $\otimes$ presents the Kronecker product. $ \mathbb{H}_{+}$ denotes the set of Hermitian semi-positive matrices.  $\arg ({\bf a})$ denotes the  argument/phase of vector ${\bf a}$.

 \section{System model}  
We consider an mmWave downlink MU-MIMO system, in which a base station (BS),  equipped with $N_\text{t}$ antennas and $N_\text{RFT}  $ radio frequency (RF)  chains, is communicating with $K$ independent users (UEs). Each UE is equipped with  $N_\text{r}$ antennas and $N_\text{RFR}  $ RF chains to receive $N_\text{s}$ data streams simultaneously. To guarantee
the effectiveness of the communication carried by the
limited number of RF chains, the number of the transmitted
streams is constrained by $KN_\text{s} \le N_\text{RFT} \le N_\text{t}$ for the BS, and  $N_\text{s} \le N_\text{RFR} \le N_\text{r}$ for each UE.
The  received signal for the $k$-th UE is given by 
\begin{equation}\label{re_signal}
{\bf y}_k = {\bf W}_{\text{BB},k}^H {\bf W}_{\text{RF},k}^H ( {\bf H}_k \sum\nolimits_{k = 1}^K   {\bf F}_\text{RF} {\bf F}_{\text{BB},k}  {\bf s}_k + {\bf n}_k  ), 
\end{equation}
where $k \in  \{1,2,...,K\}$,  ${\bf s}_k \in \mathbb{C}^{N_\text{s}}$ is the transmitted symbol vector for the $k$-th UE such that $\mathbb{E}({\bf s}_k^H{\bf s}_k)=\frac{P}{K N_\text{s}}{\bf I}_{N_\text{s}} $, $P$ is total transmit power, and ${\bf n}_k \sim  \mathcal{CN}(0,\sigma_k^2 {\bf I}_{N_\text{r}} )$ is the additive white Gaussian noise (AWGN)  at the $k$-th UE.  ${\bf F}_\text{BB,k} \in \mathbb{C}^{N_\text{RFT} \times N_\text{s}}$  and ${\bf W}_\text{BB,k} \in \mathbb{C}^{N_\text{RFR} \times N_\text{s}}$ denote the  baseband digital precoder and combiner, respectively.  
 ${\bf F}_\text{RF}  \in \mathbb{C}^{ N_\text{t} \times N_\text{RFT} }$  and ${\bf W}_{\text{RF},k}  \in \mathbb{C}^{ N_\text{r} \times N_\text{RFR} }$ are the  analog precoder and combiner for the $k$-th UE, respectively. Both  ${\bf F}_\text{RF}$  and ${\bf W}_{\text{RF},k}$ are 
  implemented  using analog phase shifters with constant modulus, i.e., $|[{\bf F}_\text{RF}]_{n,m}|=1 $ and $|[{\bf W}_{\text{RF},k}]_{n,m}|=1 $\cite{Ming2017xiao}. To meet the total transmit power constraint at the BS,  precoding matrices  ${\bf F}_\text{RF}$ and ${\bf F}_{\text{BB},k}$ are constrained by $\sum_{k=1}^{K} \| {\bf F}_\text{RF}{\bf F}_{\text{BB},k}  \|_\text{F}^2 =K N_\text{s}$. The mmWave MIMO channel between  the BS and the $k$-th UE, denoted as  ${\bf H}_k $, can be characterized by  the Saleh-Valenzuela model\cite{nguyen2017hybrid}. From the above, the   system  sum-rate  when transmitted symbols follow a Gaussian distribution is given by   
  \begin{equation}\label{sumrate}
  \mathcal{R} =\sum\nolimits_{k = 1}^K \log  |  {\bf I}_{N_\text{s}}+ {\bf F}_k^H  {\bf H}_k^H {\bf W}_k    {\bf R}_k^{-1}   {\bf W}_k^H {\bf H}_k {\bf F}_k   |,
  \end{equation}
  where ${\bf W}_k = {\bf W}_{\text{RF},k} {\bf W}_{\text{BB},k}$, ${\bf F}_k = {\bf F}_\text{RF} {\bf F}_{\text{BB},k}$, and ${\bf R}_k ={\bf W}_k^H {\bf H}_k  ( \sum_{n  \ne k  } {\bf F}_n {\bf F}_n^H  )  {\bf H}_k^H {\bf W}_k  + \rho_n {\bf W}_k^H   {\bf W}_k  $ is the covariance matrix of total inter-user interference-plus-noise at $k$-th UE and $\rho_n=\sigma_n^2 K N_\text{s}/P $. 

 \section{HBF design with FP and MM method }   
In what follows, we  maximize the achievable system   sum-rate by jointly optimizing hybrid beamformers, i.e., $\mathcal{B}_k =\{ {\bf F}_{\text{RF}}, {\bf F}_{\text{BB},k},{\bf W}_{\text{RF},k},{\bf W}_{\text{BB},k}\}$, $\forall k$.  
  The optimization problem can be stated as 
\begin{equation}\label{Problem_main}
\begin{split}
(\text{P}1):~& \mathop {\max }\limits_{\mathcal{B}_k, \forall k } ~~  \mathcal{R} \\
&~~ \text{s.t.}  ~~~   \sum\nolimits_{k = 1}^K  \| {\bf F}_\text{RF}{\bf F}_{\text{BB},k} \|_\text{F}^2 =K N_\text{s},  \\
   &~~~~~~~~~{\bf F}_\text{RF}\in \mathcal{F}_\text{RF},  {\bf W}_{\text{RF},k}\in \mathcal{W}_{\text{RF},k},\forall k,\\
\end{split}
\end{equation}  
where $ \mathcal{F}_\text{RF}$ and  $\mathcal{W}_{\text{RF},k}$  denote the feasible sets of analog beamformers which obey the constant modulus constraints for ${\bf F}_\text{RF}$ and ${\bf W}_{\text{RF},k}$. 

Obviously, the sum-rate maximization problem (P1) is non-convex and NP-hard with respect to $\mathcal{B}_k,\forall k$ due to  the coupled variables in the matrix ratio term in \eqref{sumrate} and the constant modulus constraints of  analog beamformers. To make this problem tractable, we first transform problem (P1) to an easily implemented problem according to the FP theory in  \cite{shen2019optimization}. Then, problem (P1) can be  rewritten  as  
\begin{equation}\label{Problem_main2}
\begin{split}
(\text{P}2):~ \mathop {\max }\limits_{\mathcal{B}_k, {\bf V}_k, {\bf U}_k, \forall k }   ~~&\overline{\mathcal{R}} \\
~~ \text{s.t.}  \quad~~~~ &\sum\nolimits_{k = 1}^K  \| {\bf F}_\text{RF}{\bf F}_{\text{BB},k} \|_\text{F}^2 =K N_\text{s},   \\
  &{\bf V}_k \in \mathbb{H}_{+}^{N_\text{s}\times N_\text{s}}, {\bf U}_k \in \mathbb{C}^{N_\text{s}\times N_\text{s}},  \\
 &{\bf F}_\text{RF}\in \mathcal{F}_\text{RF},  {\bf W}_{\text{RF},k}\in \mathcal{W}_{\text{RF},k}, \forall k,\\
\end{split}
\end{equation} 
where $ \overline{\mathcal{R}} = \sum_{k=1}^{K} ( \log | {\bf \Gamma}_k  | )+\text{Tr}({\bf V}_k) + 2\text{Tr}({\bf \Gamma}_k  {\bf F}_k^H {\bf H}_k^H {\bf W}_k {\bf U}_k   )- \text{Tr}(  {\bf \Gamma}_k {\bf U}_k^H \overline {\bf R}_k {\bf U}_k  )$
with  $\overline {\bf R}_k ={\bf W}_k^H {\bf H}_k ( \sum_{n=1}^K {\bf F}_n {\bf F}_n^H )  {\bf H}_k^H {\bf W}_k  + \rho_n {\bf W}_k^H   {\bf W}_k  $ and ${\bf \Gamma}_k ={\bf I}_{N_\text{s}} +  {\bf V}_k $. 
Due to the power constraint and the  non-convex constraints of analog beamformers, it is still hard to solve problem (P2) directly. Thus, we propose a  two-step  approach to solve problem (P2). In the first step, we mainly focus on maximizing $\overline{\mathcal{R}}$ by jointly optimizing the FD beamformers (i.e., $\mathcal{D}=\{{\bf F}_k, {\bf W}_k, \forall k \}$). Then, problem (P2) is reformulated as 
 \begin{equation}\label{Problem_main3}
 \begin{split}
 (\text{P}3 ):~ \mathop {\max }\limits_{\mathcal{D}, {\bf V}_k, {\bf U}_k, \forall k} ~~  &\overline{\mathcal{R}} \\
 ~~ \text{s.t.}  ~~~~~~& \sum\nolimits_{k = 1}^K  \| {\bf F}_k \|_\text{F}^2 =K N_\text{s},   \\
   &{\bf V}_k \in \mathbb{H}_{+}^{N_\text{s}\times N_\text{s}}, {\bf U}_k \in \mathbb{C}_{+}^{N_\text{s}\times N_\text{s}}. \\
 \end{split}
 \end{equation}
Note that problem (P3) is  bi-convex which can be effectively solved with  alternating optimization (AO) methods. According to AO methods, the solutions of problem (P3) can be obtained iteratively, where in iteration $i+1$, the variables are updated as follows 
\begin{subequations}\label{Digital}
	\begin{align}	
	{\bf U}_k^{(i+1)}:=& \arg \mathop{\max}\limits_{{\bf U}_k}   \overline{\mathcal{R}} \big({\bf F}_k^{(i)}, {\bf W}_k^{(i)}, {\bf V}_k^{(i)},{\bf U}_k  \big), \forall k; \label{Digital1} \\
	{\bf V}_k^{(i+1)}:=& \arg \mathop{\max}\limits_{{\bf V}_k}   \overline{\mathcal{R}} \big({\bf F}_k^{(i)}, {\bf W}_k^{(i)}, {\bf V}_k,{\bf U}_k^{(i+1)}  \big), \forall k;\label{Digital2}\\		
	{\bf W}_k^{(i+1)}:=& \arg \mathop{\max}\limits_{{\bf W}_k}   \overline{\mathcal{R}} \big({\bf F}_k^{(i)}, {\bf W}_k, {\bf V}_k^{(i+1)},{\bf U}_k^{(i+1)}  \big), \forall k;\label{Digital3}\\		
	{\bf F}_k^{(i+1)}:=& \arg \mathop{\max}\limits_{{\bf F}_k}   \overline{\mathcal{R}} \big({\bf F}_k, {\bf W}_k^{(i+1)}, {\bf V}_k^{(i+1)},{\bf U}_k^{(i+1)}  \big), \forall k.\label{Digital4}
	\end{align}	
\end{subequations}  
According to some basic differentiation rules for complex-value matrices, closed-form solutions of   problems \eqref{Digital1}-\eqref{Digital4} are correspondingly derived as 
\begin{align}
 {\bf U}_k^{(i+1)}:=&( \overline {\bf R}_k^{(i)} )^{-1} {\bf W}_k^{H(i)} {\bf H}_k {\bf F}_k^{(i)}, \forall k; \label{Digital_solution1} \\   
   {\bf V}_k^{(i+1)}:=&{\bf F}_k^{H(i)} {\bf H}_k^H {\bf W}_k^{(i)} (  {\bf R}_k^{(i)} )^{-1} {\bf W}_k^{H(i)} {\bf H}_k {\bf F}_k^{(i)}, \forall k; \label{Digital_solution2}\\
   {\bf W}_k^{(i+1)}:=& ( {\bf H}_k^H ( \sum\nolimits_{n = 1}^K {\bf F}_n^{(i)} {\bf F}_n^{H(i)}   ) {\bf H}_k  + \rho_k {\bf I}_{N_\text{t}} )^{-1}  {\bf H}_k {\bf F}_k^{(i)}  \notag \\& \times {\bf \Gamma}_k^{H(i+1)} {\bf U}_k^{H(i+1)}( {\bf U}_k^{(i+1)}  {\bf \Gamma}_k^{(i+1)}  {\bf U}_k^{H(i+1)} )^{-1}, \forall k;\label{Digital_solution3}\\
   {\bf F}_k^{(i+1)}:=&( \sum\nolimits_{n = 1}^K  {\bf H}_n^H {\bf W}_n^{(i+1)}  {\bf U}_n^{(i)}  {\bf \Gamma}_n^{(i+1)}  {\bf U}_n^{H(i+1)}  {\bf W}_n^{H(i+1)}  {\bf H}_n  \notag\\ & +\delta^*{\bf I}_{N_\text{t}} )^{-1}  {\bf H}_k^H {\bf W}_k^{(i)}  {\bf U}_k^{(i+1)}  {\bf \Gamma}_k^{(i+1)} , \forall k,\label{Digital_solution4}
\end{align}
where  the optimal multiplier $\delta^*$ is introduced for the power constraint in \eqref{Problem_main3} and it can be easily obtained by bisection search.  
After obtaining the updated variables, we summarize the FP  based FD beamforming algorithm in Algorithm \ref{ALG1}.   
\begin{algorithm}[t]\label{ALG1}
	\caption{FP  based FD beamforming algorithm} 
	\begin{algorithmic}[1]
      	\STATE \textbf{Input}: ${\bf H}_k$, $\forall k$;
      	\STATE \textbf{Output}: ${\bf F}_k, {\bf W}_k$, $\forall k$;
		\STATE \textbf{Initialize}: ${\bf F}_k^{(0)}, {\bf W}_k^{(0)}$, ${\bf V}_k^{(0)}$, $\forall k$  and  $i=0$;  
		\REPEAT
		\STATE \text{Update} ${\bf U}_k^{(i+1)}$, $\forall k$,  using \eqref{Digital_solution1};
		\STATE \text{Update} ${\bf V}_k^{(i+1)}$, $\forall k$,  using \eqref{Digital_solution2};
		\STATE \text{Update} ${\bf W}_k^{(i+1)}$, $\forall k$,  using \eqref{Digital_solution3};
		\STATE \text{Update} ~${\bf F}_k^{(i+1)}$, $\forall k$,  using \eqref{Digital_solution4};												
		\STATE $i\leftarrow i+1$;
		\UNTIL{the stopping criteria is  met.}
	\end{algorithmic} 
\end{algorithm}

With the FD beamformers derived in Algorithm 1, we then turn to optimize  hybrid beamformers $\mathcal{B}_k$,$ \forall k$.  Extensive works show  that minimizing the Euclidean distance  between the FD beamformer and the hybrid beamformer is an effective surrogate for maximizing the sum-rate   of mmWave MU-MIMO systems\cite{nguyen2017hybrid}. In what follows, we first optimize the hybrid beamformers at the BS by  minimizing the Euclidean distance  between  FD beamformers (i.e., ${\bf F}_k$, $\forall k$)  and  hybrid beamformers (i.e., $  {\bf F}_{\text{RF}} {\bf F}_{\text{BB},k}$,  $\forall k$). Letting ${\bf F}_{\text{BB}}=[{\bf F}_{\text{BB},1}, ..., {\bf F}_{\text{BB},K}]$, the problem   is formulated as  
\begin{equation}\label{Problem_HFB_1}
\begin{split}
(\text{P}4 ):~\mathop {\min }\limits_{{\bf F}_{\text{RF}}, {\bf F}_{\text{BB}} }~ & \sum\nolimits_{k = 1}^K \| {\bf F}_k -{\bf F}_{\text{RF}} {\bf F}_{\text{BB},k}\|_\text{F}^2 \\
   \quad \text{s.t.}  ~~~  &{\bf F}_\text{RF}\in \mathcal{F}_\text{RF}, \\
   \quad  \quad &   \| {\bf F}_\text{RF}{\bf F}_{\text{BB}} \|_\text{F}^2 =K N_\text{s}. \\
\end{split}
\end{equation} 
 Though problem (P4) can be solve based on OMP and manifold optimization methods, the OMP method cannot achieve high system performance and the manifold optimization method is with extremely high complexity\cite{elbir2019hybrid}. Thus, to jointly design the hybrid beamformers at the BS, we will solve problem (P4) based on the AO framework, where the analog precoder ${\bf F}_\text{RF}$ is firstly optimized by by fixing ${\bf F}_{\text{BB},k}, \forall k$.   
Hence, (P4) can be rewritten as  
\begin{equation}\label{Problem_HFB_2}
\begin{split}
(\text{P}5 ):~\mathop {\min }\limits_{{\bf F}_{\text{RF}}}~ & f({\bf F}_{\text{RF}};{\bf F}_{\text{BB}}) \\
   \quad \text{s.t.}  ~  &{\bf F}_\text{RF}\in \mathcal{F}_\text{RF}, 
\end{split}
\end{equation} 
where 
\begin{equation}\label{function_f}
\begin{split}
f({\bf F}_{\text{RF}};{\bf F}_{\text{BB}}) =& \sum\nolimits_{k = 1}^K \| {\bf F}_k -{\bf F}_{\text{RF}} {\bf F}_{\text{BB},k} \|_\text{F}^2\\
 \mathop  = \limits^{(a)} &\sum\nolimits_{k = 1}^K \text{Tr}( {\bf F}_k {\bf F}_k^H )+{\bf f}_\text{RF}^H  {\bf Q}_k {\bf f}_\text{RF} -2\text{Re}({\bf f}_\text{RF}^H  {\bf e}_k),
\end{split}
\end{equation} 
 in which (a) follows from the identity $\text{Tr}({\bf A}{\bf B}{\bf C}{\bf D})=\text{vec}({\bf A}^T)^T ( {\bf D}^T\otimes{\bf B}  )\text{vec}({\bf C})$,  ${\bf Q}_k=({\bf F}_{\text{BB},k} {\bf F}_{\text{BB},k}^H)^T \otimes {\bf I}_{N_\text{t}}$, ${\bf f}_\text{RF} = \text{vec}({\bf F}_\text{RF})$, ${\bf E}_k={\bf F}_k {\bf F}_{\text{BB},k}^H$, and ${\bf e}_k = \text{vec}({\bf E}_k)$. 
 To effectively  solve  non-convex problem (P5), we use an MM method. The basic idea is to transform  original problem (P5) into a sequence of majorized subproblems that can be solved with closed-form minimizers. At first, according to lemma 2 in \cite{wu2017transmit},   we can find a valid majorizer of $f({\bf F}_{\text{RF}};{\bf F}_{\text{BB}})$ at point ${\bf F}_\text{RF}^{(i)}\in \mathcal{F}_\text{RF}$ as
 \begin{equation}\label{Majorizer}
 \begin{split}  
    &f({\bf F}_{\text{RF}}; {\bf F}_\text{RF}^{(i)},{\bf F}_{\text{BB}})
    = \sum\nolimits_{k = 1}^K   \text{Re}( {\bf f}_\text{RF}^H ( ( {\bf Q}_k -\lambda_k{\bf I}){\bf f}_\text{RF}^{(i)} -  {\bf e}_k  )) +C, 
 \end{split}
 \end{equation} 
where $\lambda_k$ denotes the maximum eigenvalue of  ${\bf Q}_k$, and the constant term $C=\sum\nolimits_{k = 1}^K \text{Tr}( {\bf F}_k {\bf F}_k^H )+\lambda_k{\bf f}_\text{RF}^H  {\bf f}_\text{RF}+{\bf f}_\text{RF}^{H(i)}( \lambda_k{\bf I} -{\bf Q}_k  ){\bf f}_\text{RF}^{(i)}$. Then, according to the MM method and  utilizing the majorizer in \eqref{Majorizer}, the solution of problem (P4) can be obtained by  iteratively solving the following problem 
  \begin{equation}\label{Problem_HFB_3}
  \begin{split}
   (\text{P}6):~  \mathop {\min }\limits_{{\bf F}_\text{RF}} ~ & f({\bf F}_{\text{RF}}; {\bf F}_\text{RF}^{(i)},{\bf F}_{\text{BB}})
  \\
     ~~ \text{s.t.} ~~ &  {\bf F}_\text{RF}\in \mathcal{F}_\text{RF}. 
  \end{split}
  \end{equation} 
The closed-form solution of problem (P6) is given by 
\begin{equation}\label{frf}
{\bf f}_\text{RF}^{(i+1)}=-\exp (j\arg ( \frac{1}{K}\sum\nolimits_{k = 1}^K  ( {\bf Q}_k -\lambda_k{\bf I} ){\bf f}_\text{RF}^{(i)} -  {\bf e}_k ) ).
\end{equation}

Then, we turn to design digital beamformers (i.e., ${\bf F}_{\text{BB},k}, \forall k$) at the BS with fixed  ${\bf F}_\text{RF}$. By fixing ${\bf F}_\text{RF}$, the solution of problem (P4) without considering the power constraint in \eqref{Problem_HFB_1} is given by 
 \begin{equation}\label{fbb}
{\bf F}_{\text{BB},k}= {\bf F}_\text{RF}^{-1} {\bf F}_k, \forall k.  
 \end{equation} 
To satisfy the power constraint in problem (P4),  we can normalize ${\bf F}_\text{BB}$ by a factor of $\frac{\sqrt{K N_\text{s}}}{   \| {\bf F}_\text{RF} {\bf F}_\text{BB}   \|_\text{F}}$\cite{yu2016alternating}.  The effectiveness  of the normalization step is refereed to \cite{yu2016alternating,huang2020learning}. 
With above closed-form solutions in \eqref{frf} and \eqref{fbb}, we summarize the MM based HBF algorithm  in Algorithm 2. 
  \begin{algorithm}[t]\label{ALG2}
 	\caption{MM based HBF  algorithm} 
 	\begin{algorithmic}[1]
 	    \STATE \textbf{Input}: ${\bf F}_k, \forall k$;
       	\STATE \textbf{Output}: ${\bf F}_\text{RF}$, ${\bf F}_\text{BB}$; 
 		\STATE \textbf{Initialize}: ${\bf F}_\text{RF}^{(0)} $ and outer iteration $i_\text{o}=0$;   
 		\REPEAT
 		\STATE  Fix ${\bf F}_\text{RF}^{(i_\text{o})}$,  compute ${\bf F}_{\text{BB},k}, \forall k$, according   to \eqref{fbb};
 		\STATE  Use MM method to compute  ${\bf F}_\text{RF}^{(i_\text{o}+1)}$:
 			    	\STATE \textbf{Initialize}: ${\bf F}_\text{RF}^{(0)}={\bf F}_\text{RF}^{(i_\text{o})}$ and inner iteration $i_\text{i}=0$;   
 					\REPEAT 
 					\STATE   Compute ${\bf F}_\text{RF}^{(i_\text{i}+1)}$ according to \eqref{frf};
 					\STATE   $i_\text{i} \leftarrow i_\text{i}+1$;
 					\UNTIL{the stopping criteria is  met.}
 					\STATE \textbf{Update}: ${\bf F}_\text{RF}^{(i_\text{o}+1)}={\bf F}_\text{RF}^{(i_\text{i})}$;    										
 		\STATE $i_\text{o} \leftarrow i_\text{o}+1$;
 		\UNTIL{the stopping criteria is  met.}
 		\STATE Compute ${\bf F}_\text{BB}= \frac{\sqrt{K N_\text{s}}}{ \| {\bf F}_\text{RF} {\bf F}_\text{BB} \|_\text{F}  }  {\bf F}_\text{BB}$. 
 	\end{algorithmic} 
 \end{algorithm} 
The hybrid beamformers at  UEs can be designed following a similar approach to that at the BS, and details are omitted here due to space limitation.  We then summarize the  convergence and   main complexity   of Algorithm 1 and Algorithm 2 in the following theorems. 

\begin{theorem}
The convergence of Algorithm 1 is guaranteed. The main complexity of Algorithm 1 is  $\mathcal{O}(2I_\text{fp}KN_{\text{t}}^3)$, where $I_\text{fp}$ is the number of iterations.
\end{theorem}
\begin{proof}
According to \eqref{Digital}, $\{\overline{\mathcal{R}} \big({\bf F}_k^{(i)}, {\bf W}_k^{(i)}, {\bf V}_k^{(i)},{\bf U}_k^{(i)}  \big)\}$ is a monotonically non-decreasing sequence and it thus converges, since $\{\overline{\mathcal{R}} \big({\bf F}_k^{(i)}, {\bf W}_k^{(i)}, {\bf V}_k^{(i)},{\bf U}_k^{(i)}  \big)\}$ is upper bounded with power constraints.
For Algorithm 1,  the main complexity at each iteration comes from the inversion operations  in \eqref{Digital_solution1}-\eqref{Digital_solution4}, which is $\mathcal{O}(2KN_{\text{t}}^3)$.
\end{proof}

\begin{theorem}\label{theorem_alg2}
The convergence of Algorithm 2 is guaranteed. The main complexity of Algorithm 2  is $\mathcal{O}(I_\text{out} (I_\text{in} K (N_\text{t}N_\text{RFT})^3+ N_\text{t} N_\text{RFT}^2 ) ) $, where $ I_\text{out}$ and $I_\text{in}$ are the numbers of outer  and inner iterations, respectively
\end{theorem}
\begin{proof}

The proof of  convergence is similar to that in \cite{huang2020learning}, and  omitted here for space limitation.
The main  complexity at each iteration  of Algorithm 2 comes from   finding the maximum eigenvalue of ${\bf Q}_k$ and the pseudo inversion of ${\bf F}_\text{RF}$. That is $\mathcal{O}(I_\text{in} K (N_\text{t}N_\text{RFT})^3+ N_\text{t} N_\text{RFT}^2  ) $. 
\end{proof}
Above proposed HBF algorithm (denoted as FP-MM-HBF) is iterative algorithm and still suffers from high computational complexity as the number of antennas increases. 
Furthermore, since the  proposed HBF algorithm   and existing optimization based algorithms are linear mapping from the channel matrices and the hybrid beamformers, they require  a  real-time computation, and are not robust to noisy channel input data. Thus, a learning based approach to address these problems is proposed in the following section.

 \section{HBF design with ELM }    
   \begin{figure}[t] 
  	\vskip 0.2in
  	\begin{center}
  		\centerline{\includegraphics[width=85mm]{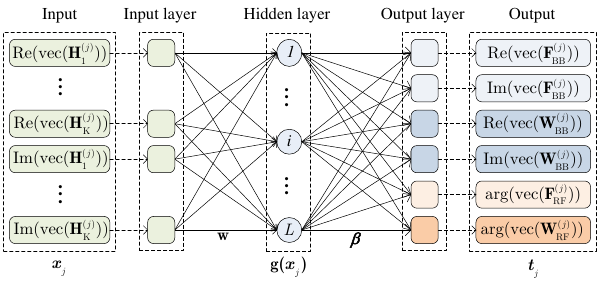}}
  		\caption{ELM network for HBF design.}
  		\label{fig_ELM}
  	\end{center}
  	\vskip -0.2in
  \end{figure}
 In what follows, we present our ELM framework for joint hybrid precoders and combiners design, shown in Fig.~\ref{fig_ELM}.  There is only one hidden layer in  ELM, the weights of input nodes and bias for the hidden nodes are generated randomly. 
    We assume that the training dataset is $\mathcal{D}= \{(\bm{x}_j,\bm{t}_j)|j=1, \ldots,N\} $, where $\bm{x}_j$ and $\bm{t}_j$ are sample and target for the $j$-th training data. 
    Specifically,  the $j$-th training data is defined as
     $\bm{x}_j=[\text{Re}(\text{vec}(\overline {\bf{H}}^{(j)}_{1})),..., \text{Re}(\text{vec}(\overline{\bf{H}}^{(j)}_{K})), \text{Im}(\text{vec}(\overline {\bf{H}}^{(j)}_{1})),...,$ $ \text{Im}(\text{vec}(\overline{\bf{H}}^{(j)}_{K}))]  \in\mathbb{R}^{N_\text{I}}$   where $\overline{\bf{H}}^{(j)}_k\sim \mathcal{CN}({\bf{H}}_k,\Gamma_k)$  and  ${N_\text{I}}= 2K N_\text{r}N_\text{t}$.
      And    $\Gamma_k $ denotes the variance  of  added synthetic noise,  with its $(m,n)$-th entry as $[\Gamma_k]_{m,n} = \frac{|[ {\bf{H}}^{(j)}_k ]_{m,n}|^2 }{10^{\text{SNR}_{\text{Train}}/20 }}$, where $\text{SNR}_{\text{Train}}$ is the SNR for the training data\cite{elbir2019hybrid}.   
    The target of $j$-th data is $\bm{t}_j=[\text{Re}(\text{vec}({\bf{F}}^{(j)}_{\text{BB}})),  \text{Im}(\text{vec}({\bf{F}}^{(j)}_{\text{BB}})) , \text{Re}(\text{vec}({\bf{W}}^{(j)}_{\text{BB}})), $ $\text{Im}(\text{vec}({\bf{W}}^{(j)}_{\text{BB}})),\arg(\text{vec}({\bf{F}}^{(j)}_{\text{RF}})),\arg(\text{vec}({\bf{W}}^{(j)}_{\text{RF}}))]  \in\mathbb{R}^{N_\text{o}}$, where 
    $N_\text{o}=N_\text{t}N_\text{RFT}+K(N_\text{r}N_\text{RFR}+ N_\text{s}(N_\text{RFR})+N_\text{RFT}))$,  ${\bf W}_{\text{BB}}=[{\bf W}_{\text{BB},1}, ..., {\bf W}_{\text{BB},K}]$, and ${\bf W}_{\text{RF}}=[{\bf W}_{\text{RF},1}, ..., {\bf W}_{\text{RF},K}]$. The beamformers in   $\bm{t}_j$ are obtained by Algorithms 1 and 2. 
       According to  \cite{ElMsurvey,yy2019}, the output of ELM related to sample $\bm{x}_j$ can be mathematically modeled as
  \begin{equation}\label{eq18}
\sum\nolimits_{l = 1}^L  \beta_l g_l(\bm{x}_j)=\sum\nolimits_{l = 1}^L \beta_l g(\bm{w}_l^T\bm{x}_j+b_l)={\bf g}(\bm{x}_j)\bm{\beta},
 \end{equation} 
 where $\bm{w}_l=[w_{l,1},\ldots,w_{l, N_\text{I}}]^T$ is the weight vector connecting the $l$-th hidden node and the input nodes, $\bm{\beta}=[\beta_1,\ldots, \beta_L]^T \in \mathbb{R}^{L\times N_\text{o}}$, and $\beta_l=[\beta_{l,1},\ldots,\beta_{l,N_\text{o}}]^T$ is the weight vector connecting the $l$-th hidden node and the output nodes, and $b_l$ is the bias of the $l$-th hidden node. 
 
 Since there is only one hidden layer in ELM, with randomized weights $\{\bm{w}_i\}$ and biases $\{b_i\}$, the goal is to tune the output weight $\bm{\beta}$ with training data $\mathcal{D}$ through minimizing the ridge regression problem 
  \begin{equation} 
  (\text{P}7):~ \bm{\beta}^*=\arg\min_{\bm{\beta}}~\frac{\lambda}{2}\| {\bf G}\bm{\beta}-{\bf T}\|^2 + \frac{1}{2} \| \bm{\beta} \|^2,
  \end{equation}
  where $ {\bf T}=[\bm{t}_{1},\ldots, \bm{t}_N]^T_{N\times N_\text{o}}$, $\lambda$ is the trade-off parameter between the training error and the regularization and 
   \begin{equation}
   {\bf G}=\left[\begin{matrix}
   {\bf g}(\bm{x}_1)\\\vdots\\ {\bf g}(\bm{x}_N)
   \end{matrix} \right]=\left[\begin{matrix}
   g_1(\bm{x}_1) & \cdots &g_L(\bm{x}_1) \\
   \vdots & \cdots & \vdots \\
   g_1(\bm{x}_N) & \cdots &g_L(\bm{x}_N)
   \end{matrix}\right]_{L\times N}.
   \end{equation}
   According to  \cite{huang2020learning}, the closed-form solution for (P7) is 
  \begin{equation}\label{beta1}
  \bm{\beta}^*={\bf G}^T(\frac{\bf{I}}{\lambda}+{\bf G}{\bf G}^T )^{-1}{\bf T} ,
  \end{equation}
 
  From above, it can be concluded that ELM is with very low complexity since there is only one layer's parameters to be trained and the weight of output layer (i.e., $\bm{\beta}$ ) is given in closed-form. 
 
\section{Numerical simulations}
In this section, we  numerically evaluate  the performance of 
  our  proposed methods, and compare them with  four  state-of-the-art methods:   BD-ZF-FD\cite{spencer2004zero}, BD-ZF-HBF\cite{ni2015hybrid}, WMMSE-OMP-HBF \cite{nguyen2017hybrid}  and CNN-HBF \cite{elbir2019hybrid}. Uniform planar array \cite{nguyen2017hybrid} is used, and the number of cluster and array for mmWave channels are set to  $5$ and $10$, respectively.
  We select $K=3$, $N_\text{s}=2$, $N_\text{r}=16$  and   $L=4000$. 
 The ELM is fed with  100 channel realizations and for each channel realization, 100 noisy channels are obtained   by adding synthetic noise with different powers of $\text{SNR}_{\text{Train}}\in \{15, 20, 25\}$ dB. 
 
\begin{figure*}[ht]
	\centering
	\subfloat[ ]{\includegraphics[width =.34\textwidth]{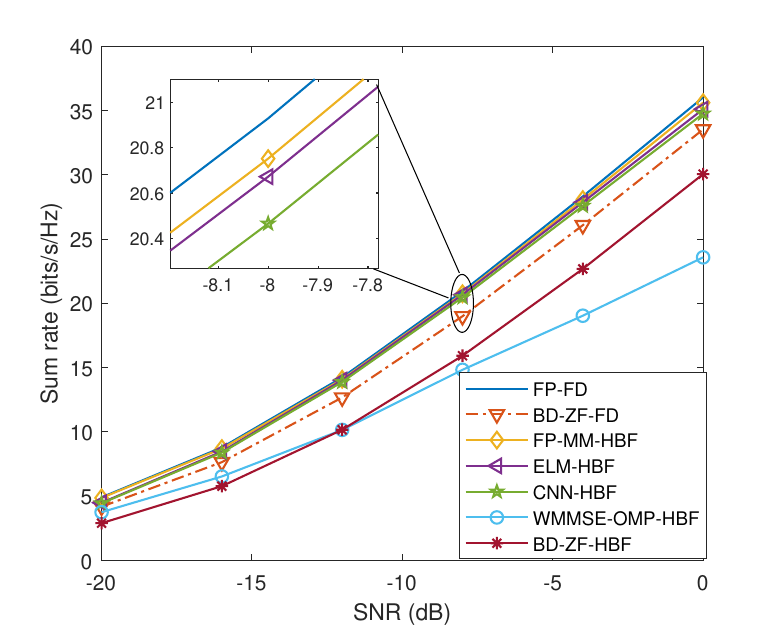}\label{fig_a}}
	\subfloat[ ]{\includegraphics[width=.34\textwidth]{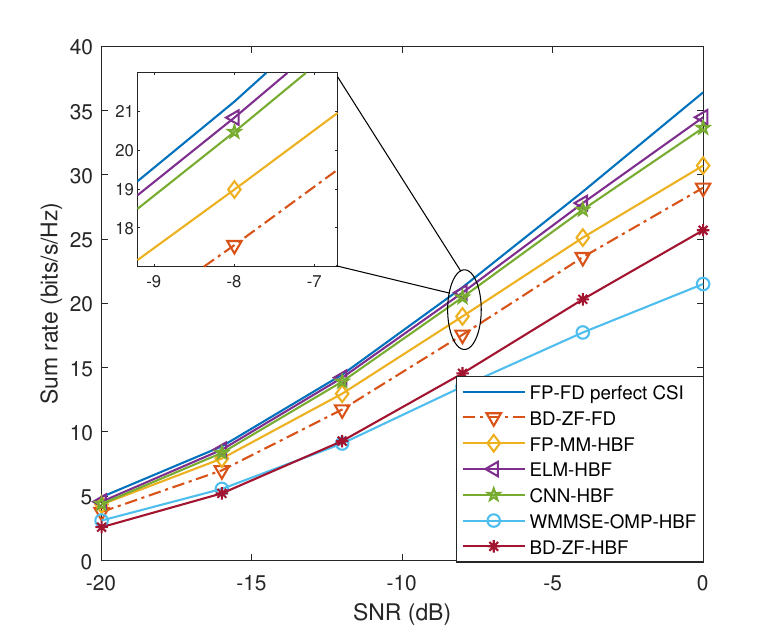}\label{fig_b}}
	\subfloat[ ]{\includegraphics[width=.34\textwidth]{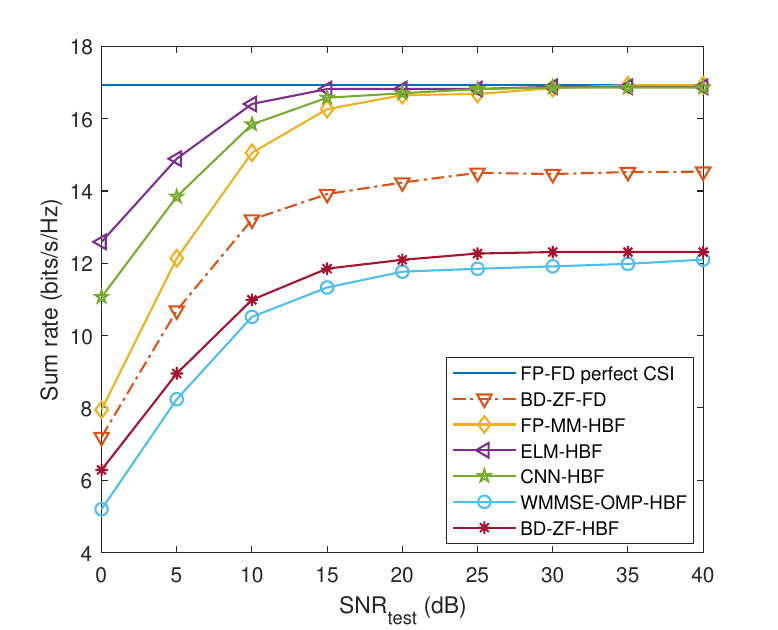}\label{fig_c}}
	\caption{Achievable sum-rate  vs SNR and $\text{SNR}_{\text{Test}}$, with  $N_t=36$,   $N_\text{RFT}=9$ and  $N_\text{RFR}=3$: (a) perfect CSI; (b) imperfect CSI (i.e., $\text{SNR}_{\text{Test}}=10$ dB); (c) sum-rate  vs  $\text{SNR}_{\text{Test}}$.}
	\label{fig_loss}
\end{figure*}

 Fig.~\ref{fig_a} and Fig.~\ref{fig_b} present the achievable sum-rate of various beamforming  methods versus SNR with perfect and imperfect CSI, respectively. With perfect CSI, it shows that all the proposed methods outperform others in Fig.~\ref{fig_a}. Both FP-MM-HBF and ELM-HBF can approach the sum-rate performance of proposed fully digital beamforming scheme, i.e., FP-FD. 
 With imperfect CSI, Fig.~\ref{fig_b} shows that learning based methods (i.e., ELM-HBF and CNN-HBF) achieve higher sum-rate than other optimization based methods. 
 We observe that ELM-HBF provides better performance than CNN-HBF and    is very close to the   performance of FP-FD. The reason is that the weight matrices of ELM are given in closed-form and  are much easier  to be optimized those that of CNN.

Table \ref{complexity_SE} shows the computation time for different HBF methods.  
The computation time of a learning based method is characterized by offline training time and online prediction time. 
A performance comparison among   two common  activation functions for ELM, i.e.,  sigmoid  function and   parametric rectified linear unit (PReLU) function\cite{huang2020learning}, is provided. 
  $1000$ noisy channel samples for $10$ channel realizations are fed into the learning machines, and $100$ noisy channel samples are used for testing. We select $\text{SNR}=-8$ dB, $N_\text{r}=16$, $\text{SNR}_{\text{Train}}=\text{SNR}_{\text{Test}}=10$ dB. We can see that 
  ELM-HBF with different activation functions  always  achieve higher sum-rate than other methods. Among all methods,   
 ELM-HBF with PReLU consumes the shortest prediction time. For instance, as $N_t=64$, it can achieve 50 and 10 times faster prediction time than WMMSE-OMP-HBF and BD-ZF-HBF, respectively.  Moreover, CNN-HBF  takes much longer training time than ELM-HBF. For instance, 
  ELM-HBF with  PReLU can achieve 1200 times faster training time than  CNN-HBF.

\section{Conclusions}
An ELM framework  is proposed to jointly design the hybrid precoders and combiners  for mmWave MU-MIMO systems. We show that proposed methods can provide higher sum-rate than existing methods. Moreover, the  proposed ELM-HBF can support more robust performance than CNN-HBF and other optimization based methods. Finally, for $N_t=64$, ELM-HBF can achieve     50   times faster prediction time than WMMSE-OMP-HBF, and  1200  times faster training time than CNN-HBF. Thus, ELM-HBF is with much lower complexity and might be more practical for implementation.

\begin{table*}[t]
\caption{Sum-rate (bits/s/Hz), Training and Prediction Time Comparison}\label{complexity_SE}
\scalebox{0.67}{
\begin{tabular}{|c|c|c|c|c|c|c|c|c|c|c|c|c|c|c|c|}
\hline
\multirow{3}{*}{Nt} & \multicolumn{2}{c|}{\multirow{2}{*}{FP-MM-HBF}}                             & \multicolumn{2}{c|}{\multirow{2}{*}{WMMSE-OMP-HBF}}                         & \multicolumn{2}{c|}{\multirow{2}{*}{BD-ZF-HBF}}                             & \multicolumn{3}{c|}{\multirow{2}{*}{CNN-HBF}}                                                                                             & \multicolumn{6}{c|}{ELM-HBF}                                                                                                                                                                                                                                                                        \\ \cline{11-16} 
                    & \multicolumn{2}{c|}{}                                                   & \multicolumn{2}{c|}{}                                                   & \multicolumn{2}{c|}{}                                                   & \multicolumn{3}{c|}{}                                                                                                                 & \multicolumn{3}{c|}{Sigmoid Node}                                                                                                                         & \multicolumn{3}{c|}{PReLU Node}                                                                                                                     \\ \cline{2-16} 
                    & \begin{tabular}[c]{@{}c@{}}Prediction\\ Time (s)\end{tabular} & Rate    & \begin{tabular}[c]{@{}c@{}}Prediction\\ Time (s)\end{tabular} & Rate    & \begin{tabular}[c]{@{}c@{}}Prediction\\ Time (s)\end{tabular} & Rate    & \begin{tabular}[c]{@{}c@{}}Prediction\\ Time (s)\end{tabular} & \begin{tabular}[c]{@{}c@{}}Training\\ Time (s)\end{tabular} & Rate    & \begin{tabular}[c]{@{}c@{}}Prediction\\ Time (s)\end{tabular} & \begin{tabular}[c]{@{}c@{}}Training\\ Time (s)\end{tabular} & Rate             & \begin{tabular}[c]{@{}c@{}}Prediction\\ Time (s)\end{tabular} & \begin{tabular}[c]{@{}c@{}}Training\\ Time (s)\end{tabular} & Rate               \\ \hline
16                  & 0.2546                                                        & 12.0939 & 0.1297                                                        & 9.0359  & 0.0410                                                        & 6.4354  & 0.0129                                                        & 1036.72                                                     & 12.5740 & 0.0480                                                        & 89.7284                                                     & 13.0325          & \textbf{0.0055}                                               & \textbf{3.1467}                                             & \textbf{13.3141} \\ \hline
36                  & 1.3227                                                        & 19.0945 & 0.3244                                                        & 13.2076 & 0.0501                                                        & 14.9696 & 0.0396                                                        & 3112.25                                                     & 20.2783 & 0.2330                                                        & 312.5433                                                    & \textbf{20.6413} & \textbf{0.0085}                                               & \textbf{3.6052}                                             & 20.3399          \\ \hline
64                  & 6.0488                                                        & 23.3880 & 0.7071                                                        & 14.9997 & 0.1362                                                        & 19.4132 & 0.0581                                                        & 5073.54                                                     & 24.3137 & 0.3041                                                        & 410.1147                                                    & 24.6988          & \textbf{0.0137}                                               & \textbf{4.3036}                                             & \textbf{24.8205} \\ \hline
\end{tabular} }
\end{table*}
 
\bibliography{ref}
\bibliographystyle{IEEEtran}

\end{document}